\documentclass[submission,copyright,creativecommons]{eptcs}
\providecommand{\event}{Non-classical logics theory \& applications 2021} 
\usepackage{breakurl}             
\usepackage{underscore} 
\usepackage{amssymb,latexsym,amsfonts,mathtools,array,enumerate,multicol,rotating,tikz,subfigure, color,url,graphicx,setspace,hhline,mathrsfs,amsthm,eucal}

\newcommand{\cnl}{\ensuremath{\mathbf{CNL_4^2}}}
\newcommand{\cnll}{\ensuremath{\mathbf{CNLL_4^2}}}

\newtheorem{theorem}{Theorem}[section]

\newtheorem{corollary}[theorem]{Corollary}

\theoremstyle{definition}
\newtheorem{definition}[theorem]{Definition}

\newtheorem{lemma}[theorem]{Lemma}

\newtheorem{proposition}[theorem]{Proposition}

\title{Cyclic Negations and Four-valuedness}
\author{Oleg Grigoriev\thanks{Oleg Grigoriev is supported by RFBR grant $\mathcal{N}$\textsuperscript{\underline{o}}\ 20-011-00698.}  \qquad\qquad Dmitry Zaitsev\thanks{Dmitry Zaitsev's research is conducted as a part of ``Brain, cognitive systems and artificial intelligence'' Lomonosov Moscow University scientific school project.}
	\institute{Lomonosov Moscow State University\\
		Faculty of Philosophy, Department of Logic,\\ Moscow, Russia}
	\email{grig@phios.msu.ru  \qquad\qquad zaitsev@philos.msu.ru}
	}
\def\titlerunning{Cyclic Negations and Four-valuedness}
\def\authorrunning{O. Grigoriev \& D. Zaitsev}
\begin{document}
	\maketitle
	
	\begin{abstract}
		We consider an example of four valued semantics partially inspired by quantum computations and negation-like operations occurred therein. In particular we consider a representation of so called square root of negation within this four valued semantics as an operation which acts like a cycling negation. We define two variants of logical matrices performing different orders over the set of truth values. Purely formal logical result of our study consists in axiomatizing the logics of defined matrices as the systems of binary consequence relation and  proving correctness and completeness theorems for these deductive systems.
	\end{abstract}

\section{Introduction}

The study of properties of negation-like connectives constitutes nowadays is a well established area of interdisciplinary research activity, including purely logical investigations (consult collective monographs~\cite{GabWan, WanFocus}). Negation often expresses the characteristic features of logical systems acting thereby as a mean for distinguishing and systematizing them (see, for instance,~\cite{odin2008} for the treatment of different types of paraconsistent logics in accordance with the properties of negations introduced there). In the literature one can find examples of hierarchic structures over the sets of negations, aimed to reflect their logical properties. Probably the best known is the ``kite of negations'' proposed by M.~Dunn in~\cite{starperp} and refined in the subsequent articles.  

In this paper, we are not intent on providing a complete picture of some big family of negation-like operations, instead we concentrate on a particular type of negation which may be characterized as a \textit{cyclic} operation over certain set of truth values. Specifically we are interested in its behaviour in the context of four-valued semantics, the breeding ground of many well known non-classical logics.

Occasionally our research was brought to life with an interest to the problematics of quantum computation and its possible representations within the semantic framework of non-classical logic. In particular the reflections on one of the most unusual quantum gates, the square root of negation, induced a unary operation on the four-element set of  truth values. On the syntactic level, we defined two logical systems considerably differing from each other with respect to the set of deductive postulates but sharing ``classicality'' of double negation. This particular feature is inherent in some other non-classical logics (\cite{humber95,kamide1,OmoriWan18,zagritwo,bifacial}).

	\section{Cyclic Negation in the Generalized Truth Values Setting}

Our interest to studies of cyclic negation stems from the different sources. This kind of negation is primarily known in the field of Post algebras and their logics (see~\cite{post,rasiowa}). Another origin can be found in the context of four valued semantics and corresponding logics. According to~\cite{Omori2015}, the first appearance of a cyclic negation in four valued framework can be found in~\cite{ruet}, while~\cite{Omori2015} itself deals with the property of functional completeness for the expansions of Belnap-Dunn logic. In particular Belnap-Dunn logic equipped with cyclic negation in~\cite{ruet} is proved to be functionally complete.  In~\cite{zagritwo}, two versions of cycling negation appeared under the names \textsl{left}  and \textsl{right turns} as the specific operations over the set of two-component generalized truth values\footnote{For the detailed account of this kind of compound truth values see~\cite{bifacial}.}, but they had not been studied there at any extent. Four-valued systems with a cyclic negation are investigated in~\cite{kamide1} and~\cite{OmoriWan18}. One of the features of the negation operations studied there consists in their ability to simulate the properties of classical (as well as intuitionistic) negation via composition. It is worth noting that~\cite{humber95} addresses the problem of simulating conventional negations via other unary operations touching upon a cyclic negation. 

\subsection{Four-Valuedness and Cyclic Negation from Quantum Computations}

Although this paper \textit{does not} concern with quantum computations or their logic at all, some concepts from the field of quantum computational logic have inspired the four-valued semantics underlying the logics discussed below and, specifically, the choice of the unary operation acting over there. This section clarifies the origins of the family of truth values used below.

One of the ideas that motivated this paper, namely, to merge generalized truth values approach and quantum computation in a joint logical framework, was prompted by seminal writings of prominent logicians of past and present, and after all is connected with the search of answers to the question, what (modern) logic is. 

The first one was proposed by G.~Frege and J.~\L{}ukasiewicz many years ago and now enjoys a new lease on life within the project of generalized truth values. The core idea may be expressed in \L{}ukasiewicz's words -- logic is the science of objects of a special kind, namely a science of logical values. Though seems strange, this understanding of logic is coherent with standard conception of logic, because the search for criteria of correct reasoning and argument immediately leads one to truth-(or, designated value-) preserving interpretation of logical inference. 

Another conception of logic is due to J. van Benthem, who in \cite{vanben2011} develops a  program of Logical Dynamics, which presupposes the interpretation of logic as a theory of information-driven agency, being thus the study of explicit informational processes (inference, observation, communication). The latter interpretation may be seen as the other side of the same coin~-- in words of J.~van Benthem, ``inference is just one way of producing information, at best on a par, even for logic itself, with others'' \cite[p. 183]{vanben2008}, so it is little wonder that ``inference and information update are intertwined'' \cite[p. 189]{vanben2008}.

One step away from here and just a moment to go, there is an idea to consider quantum logic as logic of quantum computation, where the latter offers a new possibility opened up by quantum gates to deal with information processing procedures being generalizations of reasoning and argument. An additional interest is connected with logical formalization of so called genuine quantum gates ``that transform classical registers into quregisters that are superpositions: the square root of the negation and the square root of the identity'' \cite[p. 298]{dalla2005}. According to \cite{deutsch} ``logicians are now entitled to propose a new logical operation $\sqrt{\mathtt{NOT}}$. Why? Because a faithful physical model for it exists in nature''.  

Let us remind some key concepts of quantum computational logic (for more details see, for example,~\cite{dalla2013}). The unit of representation of quantum information is a qubit (from English ``quantum bit''), $a|0\rangle +b|1\rangle$, where $|0\rangle$ and $|1\rangle$ are vectors $\begin{pmatrix}
1\\0 \end{pmatrix}$ and $\begin{pmatrix}0\\1\end{pmatrix}$, respectively, written in so called Dirac notation, while $a$ and $b$ are complex numbers, the \textit{amplitudes}, expressing the  probabilities.

 Quantum computational logic offers a broad family of operators, quantum logic gates\footnote{Well known examples are \textsc{cnot, toffoli, fredkin, swap} gates which perform reversible computation using some qubits as control registers for governing the actions on target bit. For example, \textsc{cnot} negates its target bit if and only if the control bit is recognized as 1.}, which in some cases can be rendered as the counterparts of classical logic gates and thus give rise to a family of propositional connectives in formal languages of quantum logical systems. But quantum computations provide also examples of non-classical gates. The square root of negation is of the special interest for us. For a qubit $|\varphi\rangle=a|0\rangle +b|1\rangle$, $\sqrt{\mathtt{NOT}}(|\varphi\rangle) = \frac{1}{2}[(1+i)a+(1-i)b]|0\rangle+\frac{1}{2}[(1-i)a+(1+i)b]|1\rangle$, where $i$ is an imaginary unit. 
 While $\mathtt{NOT}$ gate transforms $|1\rangle$ into $|0\rangle$ and vice versa, $\sqrt{\mathtt{NOT}}$ does only half of the work.
 
 The key observation here is that the square root of the negation is a kind of ``connective with memory''. In particular, when applied twice to Truth, it returns Falsity and vice versa. At the same time, the first application to True or False gives intermediate value. Thus, to understand where to go after the first application of the square root of the negation, one should somehow remember the point of departure. The complex nature of generalized truth values allows to yield this peculiarity by preserving the component of the initial value. For example, starting with $\bf T $, the first application of the square root of the negation ``adds'' uncertainty thus producing $ \bf{TU} $; the second application transforms it to $\bf F $; the third again adds $ \bf U $ to $\bf F $ resulting in $ \bf{FU} $; and finally after the fourth application we arrive at $\bf T $. So we can see that our representation of the square root of negation within four-valued framework is nothing more then \textsl{a cyclic negation}.
 
 Thus we have new set of truth values, $\{\mathbf{T}, \mathbf{TU}, \mathbf{FU},\mathbf{ F}\}$, and an open choice of order relation and subset of the designated values. Below we consider two natural variants of partial order over this set with the same two-element subset of designated values, $\{\mathbf T,\mathbf{TU}\}$. The choice of this subset seems reasonable for several reasons. It contains Truth itself (\textbf{T}) and the the other value (\textbf{TU}), having something that we would call \textsl{a trace of truth}. Moreover, this subset is one of the two prime filters in lattice $4\mathcal{Q}$ described below.
 
  In this paper, we consider two propositional logics, \cnl\ and \cnll,\ of four-valued matrices  (with two-valued matrix filters) constructed over the set of generalized truth values inspired by quantum computations as explained above. Though these logics have much in common, they differ essentially with respect to the properties of negations and their interrelation with conjunction and disjunction. 
 
 \subsection{Four-Valued Matrices}
 
 For both logics, \cnl\ and \cnll, we subsume the same propositional language $\mathcal{L}$ of signature\linebreak $\{\wedge, \vee, \neg\}$ over denumerable set of variables $PV$ with the set of complex formulas $FM$ constructed according to the standard inductive definition. 
 
 On the basis of the set $\mathscr U = \{\mathbf{T}, \mathbf{TU}, \mathbf{FU},\mathbf{ F}\}$ we define two distinct matrices, $\mathcal M^{\cnl}$ and $\mathcal M^{\cnll}$, over this set with the same subset of designated values $\mathcal D=\{\mathbf{T}, \mathbf{TU}\}$ and the same definition of unary operation $\mathcal O=\{{\sim},\wedge,\vee\}$ differing  with respect to meet and join in the lattice reducts of these matrices.
 
 Tableau definitions for the binary operations $ \wedge $ and $ \vee $
 can be easily imported from the order relations over the set of truth values represented via Hasse diagrams, depicted in Figure~\ref{lattices}.  Evidently these ordered sets of truth values constitute two simple lattices, $4\mathcal Q$ (left diagram) and  $4\mathcal{LQ}$.
 
 \begin{definition}
 	$\mathcal M^{\cnl}$ matrix is a structure $\langle\mathscr U, \{f_{c}\}_{c\in \mathcal O}, \mathcal D \rangle$, where the operations $f_{\wedge}$ and $f_{\vee}$ are defined as meet and join in $4\mathcal Q$, $f_{\sim}$ is defined via the following table:
 	\begin{center}
 		\begin{tabular}{|c||c|}
 			\hline
 			$ x $ & $ f_{\sim}(x) $  \\
 			\hline
 			\hline
 			$\mathbf{T}$ & $\mathbf{TU}$  \\
 			\hline
 			$\mathbf{TU}$ & $\mathbf{F}$ \\
 			\hline
 			$\mathbf{F}$ & $\mathbf{FU}$ \\
 			\hline
 			$\mathbf{FU}$ & $\mathbf{T}$ \\
 			\hline
 		\end{tabular}
 	\end{center}
 \end{definition}
 
 \begin{definition}
 	$\mathcal M^{\cnll}$ matrix is a structure $\langle\mathscr U, \{g_c\}_{c\in\mathcal O}, \mathcal D\rangle$,  where the operations $g_{\wedge}$ and $g_{\vee}$ are defined as meet and join in $4\mathcal{LQ}$, $g_{\sim}$ is defined via the same table as $f_{\sim}$.
 \end{definition}
 A valuation $v$ is a mapping $PV\mapsto\mathscr U$. An extension of $v$ to the set $FM$ depends on a matrix assumed. For example, in case of $\mathcal M^{\cnl}$ we have the following expressions for all $A, B\in FM$: $v(A\wedge B) = f_{\wedge}(v(A), v(B))$, $v(A\vee B) = f_{\vee}(v(A), v(B))$, $v(\neg A)= f_{\sim}(v(A))$. Thus we have to distinguish $\cnl$- and $\cnll$-valuations\footnote{We will avoid cumbersome subscripts like in $v_{\cnl}$ when possible.}.
 
 The semantic consequence relation is defined via preservation of a designated truth value and again relies on a matrix assumed:
 
 \begin{definition}
 	For all $A,B\in FM$, 
 	\begin{enumerate}
 		\item $A\vDash_{\cnl} B\Leftrightarrow v(A)\in D\Rightarrow v(B)\in D$, for each \cnl-valuation $v$,
 		\item $A\vDash_{\cnll} B\Leftrightarrow v(A)\in D\Rightarrow v(B)\in D$, for each \cnll-valuation $v$.
 	\end{enumerate}
 \end{definition}
 
 \begin{figure}
 	\begin{center} 
 		\begin{tikzpicture}
 		\node (B) at (0,0) {\footnotesize{\textbf{T}}};
 		\node  (F) at (-1,-1)  {\footnotesize{\textbf{TU}}};
 		\node (T) at (1,-1) {\footnotesize{\textbf{FU}}};
 		\node (N) at (0,-2) {\footnotesize{\textbf{F}}};
 		\draw (N) -- (F) -- (B);
 		\draw (N) -- (T) -- (B) ;
 		\end{tikzpicture}
 		\qquad
 		\begin{tikzpicture}
 		\node (T) at (0,-1) {\footnotesize{\textbf{T}}};
 		\node  (TU) at (0,-2)  {\footnotesize{\textbf{TU}}};
 		\node (F) at (0,-3) {\footnotesize{\textbf{F}}};
 		\node (FU) at (0,-4) {\footnotesize{\textbf{FU}}};
 		\draw (T) -- (TU) -- (F) -- (FU) ;
 		\end{tikzpicture}
 		\label{lattices}
 	\end{center}
 	\caption{Lattices $4\mathcal{Q}$ and $4\mathcal{LQ}$.} 
 \end{figure}

 It is instructive to examine set $\mathscr U$ from the generalized truth values perspective. A common way to construct a set of generalized truth values is to get powerset over some semantic basis. So, let us choose the basic set $\{\mathbf T, \mathbf U \}$, consisting of Truth and Uncertainty values, obtaining thereby the set of generalized truth values  $\{\{\mathbf{T, U} \}, \{\mathbf{T}\}, \{\mathbf{U}\},\varnothing \}$. It is natural to think of $\{\mathbf{T}\}$ as just $\mathbf T$, while $\{\mathbf{T, U}\}$ as our $\mathbf{TU}$. Then $\mathbf U$ is just ``uncertainty without being true''. Recall that the absence of truth can be understood as just being false. This suggests that $\mathbf U$ can be thought as $\mathbf{FU}$; likewise $\varnothing$ is just $\mathbf F$.

	 \section{Binary consequence systems for \cnl\ and \cnll}\label{axiomatic}
	
	To formalize semantically defined consequence relation we will use a specific variant of a logical calculus, ``a binary consequence system''\footnote{See~\cite[Chapter 6]{dunnhardegree} for a discussion of terminology concerning to different presentations of logical systems. In particular our approach is called ``binary implicational system'' there.}, which is typical of all \textbf{FDE}-related logics. The term ``binary'' means that a sequent\footnote{We use the term `sequent' in a broad sense, not reffering here to the apparatus of Gentzen calculi.} is an expression of a form $A\vdash B$ which contains exactly one formula in the antecedent or consequent position. We take some \textit{schemata} of sequents regarded as the axiomatic schemata. A sequent is an axiom if it is a particular instance of a schema. To make the presentation succinct we abbreviate ${\sim}{\sim}$ as ${\sim}^{\emph{\tiny 2}}$, ${\sim}{\sim}{\sim}$ as ${\sim}^{\emph{\tiny 3}}$ and so on. 
	
	\begin{definition}
		A sequent $A\vdash B$ is called \cnl-valid (\cnll-valid) $\Leftrightarrow$
		$$ A\vDash_{\cnl}B \quad (A\vDash_{\cnll}B). $$
		
	\end{definition}
	\begin{definition}
		A \cnl-proof (a \cnll-proof) as a list of sequents each of them is whether an axiom of \cnl (an axiom of \cnll) or derived from the previous items of the list using some rule of inference. A \cnl-proof (\cnll-proof) for a sequent $A\vdash B$ is a \cnl-proof (\cnll-proof) the last item of which coincides with $A\vdash B$. A sequent $A\vdash B$ is called \cnl-provable (\cnll-provable) if there is a \cnl-proof (\cnll-proof) for $A\vdash B$. 
	\end{definition} 	 
	
	To indicate that a sequent $A\vdash B$ is \cnl-provable (\cnll-provable) we also adopt the expression $A\vdash_{\cnl} B$ ($A\vdash_{\cnll}B$).
	
	\medskip
	{\cnl\ \& \cnll\ {\sc common axiomatic schemata and rules of inference}}:
	\begin{enumerate}[({a}1)]\itemsep=0pt
		\begin{multicols}{2}
			\item $A\wedge B\vdash A$,
			\item $A\wedge B\vdash B$,
			\item $B\vdash A\vee B$,
			\item $A\vdash A\vee B$,
			\item $A\vdash{\sim}^{\emph{\tiny 4}}A$,
			\item $ {\sim}^{\emph{\tiny 4}}A\vdash A $,
			\item ${\sim}A\wedge{\sim}B\vdash{\sim}(A\wedge B)$,
			\item ${\sim}(A\vee B) \vdash {\sim}A\vee{\sim}B$,
			\item $A\wedge{\sim}^{\emph{\tiny 2}}A\vdash B$,
			\item $A\wedge(B\vee C)\vdash (A\wedge B)\vee(A\wedge C)$.
		\end{multicols}
	\end{enumerate}
	
	\begin{enumerate}[({r}1)]\itemsep=0pt	
		\begin{multicols}{2}
			\item $A\vdash B$, $B\vdash C$\,/ $A\vdash C$,
			\item  $A\vdash B$, $A\vdash C$\,/ $A\vdash B\wedge C$,	
			\item  $A\vdash C$, $B\vdash C$\,/ $A\vee B\vdash C$,
			\item 	$A\vdash B$\,/ ${\sim}^{\emph{\tiny 2}}B\vdash{\sim}^{\emph{\tiny 2}}A$.
		\end{multicols} 
	\end{enumerate}
	
	\medskip
	{\cnl\ \textsc{additional axiomatic schemata}}:
	\begin{enumerate}[({b}1)]\itemsep=0pt
		\begin{multicols}{2}
			\item ${\sim}(A\wedge B)\vdash {\sim}A\wedge{\sim}B$,
			\item ${\sim}A\vee{\sim}B\vdash{\sim}(A\vee B)$.
		\end{multicols}
	\end{enumerate}
	
	\medskip
	{\cnll\ \textsc{additional axiomatic schemata}}:
	\begin{enumerate}[({c}1)]\itemsep=0pt
		\begin{multicols}{2}
			\item ${\sim}A\wedge{\sim}B\vdash{\sim}(A\vee B)$,
			\item ${\sim}(A\wedge B)\vdash{\sim}A\vee{\sim}B$,
			\item ${\sim}A\wedge{\sim^{\emph{\tiny 2}}}A\vdash {\sim}(A\wedge B)$
			\item $A\wedge{\sim}A\vdash{\sim}(A\vee B)$,
			\item ${\sim}(A\vee B)\vdash{\sim}A\vee B$,
			\item ${\sim}(A\vee B)\vdash{\sim}(B\vee A)$,
			\item ${\sim}(A\wedge B)\vdash{\sim}(B\wedge A)$,
			\item $({\sim}(A\vee B)\wedge{\sim}(A\wedge B))\vdash{\sim}A\wedge{\sim}B$,
		\end{multicols}
	\end{enumerate}
	
	\medskip
	
	\begin{proposition}\label{halfdm}The following sequents are provable in \cnl:
		\begin{enumerate}[(1)]\itemsep=0pt
			\item ${\sim}A\wedge{\sim}B\vdash{\sim}(A\vee B)$,
			\item ${\sim}(A\wedge B)\vdash{\sim}A\vee{\sim}B$.
		\end{enumerate}
	\end{proposition}
	\begin{proposition}\label{derivable2}
		The following sequents are provable in both \cnl\ and \cnll.
		\begin{tabbing}
			(\sc De1) ${\sim}^{\emph{\tiny 2}}A\wedge{\sim}^{\emph{\tiny 2}}B\dashv\vdash{\sim^{\emph{\tiny 2}}}(A\vee B)$,\\
			(\sc De2) ${\sim}^{\emph{\tiny 2}}A\vee{\sim}^{\emph{\tiny 2}}B\dashv\vdash{\sim^{\emph{\tiny 2}}}(A\wedge B)$,\\
			(\sc T) $B\vdash A\vee{\sim}^{\emph{\tiny 2}}A$.
		\end{tabbing}
	\end{proposition}

Systems \cnl\ and \cnll\ have much in common with classical logic. Indeed, if we were intended to represent classical logic as a binary consequence system, we would take (a1)--(a6), (a10) and (r1)--(r4), adding paradoxical postulates like (A9) (then, of course, a pair $ {\sim\sim} $ should be treated as classical $ \neg $). Is is well known that an alternative formulation of classical system is obtained by replacing contraposition rule with a full collection of De Morgan laws (but then both $ A\wedge{\neg}A\vdash_C B $ and $ A\vdash_C B\vee{\neg}B $  are needed, where $ \vdash_C $ stands for classical binary consequence relation) as axiomatic schemas. For further references we will denote this system as $C$. 

As mentioned above, double $\sim$ have all these features of classical negation. Thus a kind of \textit{intrinsic} classicality present in both our systems. More precisely we can represent this fact via translation function $\Phi$ from the language of classical logic $\mathcal{LC}$ (over the signature $\{\wedge,\vee,\neg\}$) to the language of the present systems (both sharing the same set of variables $PV$):
\begin{align*}
\Phi(p)&=p,\quad p\in PV,\\
\Phi(A\circ B)&=\Phi(A)\circ\Phi(B),\quad \circ\in\{\wedge,\vee\},\\
\Phi(\neg A)&={\sim}{\sim}\Phi(A).
\end{align*}

We would like to show, that $\Phi$ is not only a translation, but an embedding function as well. We prove this statement via semantic argument. Let us consider an expression $A\vDash_C B$ as an assertion about classical consequence relation with respect to four-valued boolean algebra based on the lattice $4\mathcal Q$.

\begin{lemma}
	For all formulas $A$, $B$ of the language $\mathcal{LC}$:
	\[
	A\vDash_C B\Leftrightarrow \Phi(A)\vDash_{\cnl}\Phi(B)\Leftrightarrow \Phi(A)\vDash_{\cnll}\Phi(B).
	\]
\end{lemma}
\begin{proof}
	Recall that a valuation is a mapping $PV\mapsto\mathscr U$, but for the extended valuations we should distinguish mappings depending on the languages and underlying semantic structures. Let us suppose for the purposes of this lemma that given a valuation $v$, we denote by $v_1$, $v_2$ and $v_3$ its classical, \cnl- and \cnll-extensions respectively.
	
	Induction on the structure of a formula shows that for all valuation $v$ and all its extensions  $v_1$, $v_2$, $v_3$ and any formula $A$ of the language $\mathcal{LC}$, 
	\begin{equation}\label{k1}
	v_1(A)=v_2(\Phi(A))=v_3(\Phi(A)).
	\end{equation}

	Next we define partial clockwise ($v^+$) and counter-clockwise ($ v^- $) rotation of a valuation $v$ for an arbitrary $p\in PV$:
	\begin{center}
		\begin{tabular}{|c||c|c|}
			\hline
			$v(p)$ & $v^+(p)$ & $ v^-(p) $\\
			\hline\hline
			$\mathbf{T}$ & $\mathbf{T}$ & $\mathbf{T}$\\
			\hline
			$\mathbf{TU}$ & $\mathbf{T}$ & $\mathbf{F}$\\
			\hline
			$\mathbf{F}$ & $\mathbf{F}$ & $\mathbf{F}$\\
			\hline
			$\mathbf{FU}$ & $\mathbf{F}$ & $\mathbf{T}$\\
			\hline
		\end{tabular}
	\end{center}
	
	Routine check  proves the following key fact, namely that for a valuation $v$, partial clockwise rotation $v^+$ induces corresponding rotations $v_i^+$, $i\in\{1,2,3\}$, while $v^-$ induces $v_i^-$, $i\in\{1,2,3\}$, answering the description given in the table below for each formula $ A\in \mathcal{LC}$:
	\begin{center}
		\begin{tabular}{|c||c|c|}
			\hline
		 	$v_1(A)$, $ v_2(\Phi(A)) $, $ v_3(\Phi(A)) $ & $v^+_1(A)$, $ v^+_2(\Phi(A)) $, $ v^+_3(\Phi(A)) $ & $v^-_1(A)$, $ v^-_2(\Phi(A)) $, $ v^-_3(\Phi(A)) $ \\
			\hline\hline
			$\mathbf{T}$ & $\mathbf{T}$ & $\mathbf{T}$\\
			\hline
			$\mathbf{TU}$ & $\mathbf{T}$ & $\mathbf{F}$ \\
			\hline
			$\mathbf{F}$ & $\mathbf{F}$ & $\mathbf{F}$ \\
			\hline
			$\mathbf{FU}$ & $\mathbf{F}$ & $\mathbf{T}$ \\
			\hline
		\end{tabular}\quad($\star$)
	\end{center}
		
	Now suppose $A\vDash_C B$, but $\Phi(A)\nvDash_{\cnl}\Phi(B)$. This means that there is some valuation $v_2$ such that $v_2(\Phi(A))\in D$, $v_2(\Phi(B))\notin D$. Assume $v_2(\Phi(A))=\mathbf{T}$, $v_2(\Phi(B))\neq\mathbf{T}$. Then, by~\eqref{k1}, $v_1(A)=\mathbf{T}$, $v_1(B)\neq\mathbf{T}$. Hence $A\nvDash_{C}B$. 
	
	Next assume $v_2(\Phi(A))=\mathbf{TU}$, but $v_2(\Phi(B))=\mathbf{FU}$ or $v_2(\Phi(B))=\mathbf{F}$.
	Then, according to table ($ \star $), $v_2^+(\Phi(A))=v^+_1(A)=\mathbf{T}$, but $v_2^+(\Phi(B))=v^+_1(B)\neq\mathbf{T}$, so $A\nvDash_{C}B$. 
	
	For the other direction assume $A\nvDash_{C}B$, that is $v_1(A)=\mathbf{T}$, $v_1(B)\neq\mathbf{T}$. According to~\eqref{k1}, this means that $v_2(\Phi(A))=\mathbf{T}$, $v_2(\Phi(B))\neq\mathbf{T}$. If $v_2(\Phi(B))\neq\mathbf{TU}$, then $\Phi(A)\nvDash_{\cnl}\Phi(B)$. Consider, however, the case $v_2(\Phi(B))=\mathbf{TU}$. According to ($\star$) we can take the rotation $v_2^-$ such that $v_2^-(\Phi(A))=\mathbf{T}$ and $v_2^-(\Phi(B))=\mathbf{F}$. Hence, again, $\Phi(A)\nvDash_{\cnl}\Phi(B)$.
	
	An analogues provides the proof in case of $\vDash_{\cnll}$ relation. 
\end{proof}
\begin{corollary}
	$\Phi$ is an embedding function.
\end{corollary}
\section{Soundness and Completeness of \cnl}
\subsection{Soundness}

\begin{lemma}[{\sc Local Soundness for \cnl}]\label{locsoundcnl}
	All axiomatic schemata of \cnl\ represent \cnl-valid sequents and the rules of inference preserve \cnl-validity.
\end{lemma}
\begin{proof} 
	We need to check each item from the list of axiomatic schemata and inference rules. Let us provide a couple of cases as an illustration.
	
	Suppose that axiomatic schemata (a9) is invalid, i.\,e. there a \cnl-valuation  $\upsilon$ that $\upsilon(A\wedge{\sim}^{\emph{\tiny 2}}A)\in\{\textbf{T}, \textbf{TU}\}$ and $\upsilon(B)\not\in\{\textbf{T}, \textbf{TU}\}$ that is $\upsilon(B)\in\{\textbf{F}, \textbf{FU}\}$. We immediately derive a contradiction since $A\wedge{\sim}^{\emph{\tiny 2}}A$ cannot take its value from the set $\{\mathbf T, \mathbf{TU}\}$ at all.
	
	Suppose that the rule (r4) does not preserve validity. This means that there is such valuation that $A\vDash_{\cnl} B$, but ${\sim}^{\emph{\tiny 2}}B\not\vDash_{\cnl}{\sim}^{\emph{\tiny 2}}A$. From the latter it follows that there is a \cnl-valuation $\upsilon$ such that $\upsilon({\sim}^{\emph{\tiny 2}}B)\in\{\textbf{T}, \textbf{TU}\}$ and $\upsilon({\sim}^{\emph{\tiny 2}}A)\not\in\{\textbf{T}, \textbf{TU}\}$ which means that $\upsilon({\sim}^{\emph{\tiny 2}}A)\in\{\textbf{F}, \textbf{FU}\}$. It is easy to note that it leads us to $\upsilon(A)\in\{\textbf{T}, \textbf{TU}\}$ and $\upsilon(B)\in\{\textbf{F}, \textbf{FU}\}$, but this contradicts with $A\vDash_{\cnl} B$, because this must be that $\upsilon(B)\in\{\textbf{T}, \textbf{TU}\}$. Therefore, (r4) preserves validity.

	The other cases are similar.
\end{proof} 
\begin{theorem}[{\sc Soundness for \cnl}]
	For any formulas $A$ and $B$, the following holds
	\begin{center} 
		$A\vdash B\text{ is \cnl-provable }\Rightarrow A\vDash_{\cnl} B$.
	\end{center} 
\end{theorem}	
\begin{proof}
	By induction on the length of the proof, using Lemma \ref{locsoundcnl}.
\end{proof}
\subsection{Completeness}
The idea of the completeness proof is based on a technique elaborated by J.\,M.~Dunn for the system of $\mathbf{FDE}$ (see \cite{dunn00}). This method essentially relies on the notion of a prime theory which is presented by the following definition.
\begin{definition}
	A \cnl-theory is the set of formulas $\alpha$ such that for all formulas $A$ and $B$,
	\begin{enumerate}\itemsep=0pt
		\item $A\wedge B\in\alpha$ whenever $A\in\alpha$ and $B\in\alpha$,
		\item $B\in\alpha$ whenever $A\in\alpha$ and $A\vdash B$ is \cnl-provable.
	\end{enumerate}
	A \cnl-theory is prime if $A\vee B\in\alpha$ implies $A\in\alpha$ or $B\in\alpha$.
	We call a \cnl-theory $\alpha$ \textit{c-normal} when for each formula $A$ it holds that $A\in\alpha$ if and only if ${\sim}^{\emph{\tiny 2}}A\notin\alpha$.
\end{definition}

As a first step toward completeness theorems for \cnl\ we prove the Extension Lemma. Note that we use this lemma uniformly for both completeness theorems. So we prove it for the case of \cnl, while proof for another system is the same.
\begin{lemma}[{\sc Extension Lemma}]\label{extension}
	For all formulas $A$ and $B$, if $A\vdash B$ is not \cnl-provable, then there is a c-normal prime theory $\alpha$ such that $A\in\alpha$, $B\not\in\alpha$.
\end{lemma}
\begin{proof}
	Suppose that for some formulas $A$ and $B$, $A\vdash B$ is not \cnl-provable. Let us define $\alpha_{0}=\{C \mid A\vdash_{\cnl} C\}$. $\alpha_0$ is a theory as it is closed under $\vdash_{\cnl}$ and $\wedge$ (using the rule (r2)). Next we construct the sequence of theories taking some enumeration of the set $FM$ ($A_1, A_2, \ldots$)  and define
	$$
	\alpha_{n+1}=
	\begin{cases}
	\alpha_n, \text{ if } \alpha_{n}\cup\{A_{n+1}\} \vdash_{\cnl} B,\\
	\alpha_{n}\cup\{A_{n+1}\}, \text{ if } \alpha_{n}\cup\{A_{n+1}\}\not\vdash_{\cnl} B.
	\end{cases}
	$$
	
	Let $\alpha$ be the union of all $\alpha_n$'s. First we show that $\alpha$ is a prime theory such that $A\in\alpha$ and  $B\not\in\alpha$. $A\in\alpha$ by the construction of $ \alpha $. Assume $B\in\alpha$, hence $B$ was added to $\alpha_i$ on $i$-th stage  of construction of the sequence, which is impossible. For the primeness suppose that $\alpha$ is not prime, i.\,e. $C\vee D\in\alpha$, but $C\not\in\alpha$ and $D\not\in\alpha$. This means that both extensions $\alpha\cup\{C\}$ and $\alpha\cup\{D\}$ contain $B$. Then there is a conjunctions of formulas form $\alpha$, say $E$, such that $E\wedge C\vdash_{\cnl} B$ and $E\wedge D\vdash_{\cnl} B$. From this, using $(r3)$, we derive $(E\wedge C)\vee(E\wedge D)\vdash_{\cnl} B$. Then, using (a10) and $(r1)$, we have $E\wedge(C\vee D)\vdash_{\cnl} B$, so $B\in\alpha$. 
	
	Finally, $\alpha$ is also c-normal. Indeed, if for some $k$, $A_k\in\alpha$ and ${\sim}^{\emph{\tiny 2}}A_k\in\alpha$, then there is an $\alpha_i$ which contains $A_k\wedge{\sim}^{\emph{\tiny 2}}A_k$ as well as $B$, due to axiom schema $A\wedge{\sim}^{\emph{\tiny 2}}A\vdash B$ and the rules (r1) and (r2), contrary to the assumption. On the other hand, primeness of $\alpha$ and derivable schema $B\vdash_{\cnl} A\vee{\sim}^{\emph{\tiny 2}}A$ guarantee that for each $A_k$, one of two formulas, $A_k$ and ${\sim}^{\emph{\tiny 2}}A_k$, belongs to $\alpha$. 
\end{proof}    

\subsection{Completeness for \cnl}

Let $\mathcal{A}$ denote the set $\{\textbf{TU}, \textbf{F}\}$. We can express our truth-values in terms of $\mathcal A$ and $\mathcal D$ sets via the following expressions:
\begin{align*}
&v(A)=\textbf{T}\text{ iff }v(A)\in\mathcal{D}\text{ and }v(A)\notin\mathcal{A},\\
&v(A)=\textbf{TU}\text{ iff }v(A)\in\mathcal{D}\text{ and }v(A)\in\mathcal{A},\\
&v(A)=\textbf{F}\text{ iff }v(A)\notin\mathcal{D}\text{ and }v(A)\in\mathcal{A},\\
&v(A)=\textbf{FU}\text{ iff }v(A)\notin\mathcal{D}\text{ and }v(A)\notin\mathcal{A}.
\end{align*}
It is not difficult to see the next lemma, having in mind the interpretations of propositional connectives.
\begin{lemma}\label{eq}
	Let $A, B\in FM$, and $v$ be a \cnl-valuation. Then, the following expressions hold:
	\begin{enumerate}[(1)]\itemsep=0pt
		\item $v({\sim}A)\in\mathcal{D}$ iff $v(A)\notin\mathcal{A}$;
		\item $v({\sim}A)\in\mathcal{A}$ iff $v(A)\in\mathcal{D}$;
		\item $v(A\wedge B)\in\mathcal{D}$ iff $v(A)\in\mathcal{D}$ and $v(B)\in\mathcal{D}$;
		\item $v(A\wedge B)\in\mathcal{A}$ iff $v(A)\in\mathcal{A}$ or $v(B)\in\mathcal{A}$;
		\item $v(A\vee B)\in\mathcal{D}$ iff $v(A)\in\mathcal{D}$ or $v(B)\in\mathcal{D}$;
		\item $v(A\vee B)\in\mathcal{A}$ iff $v(A)\in\mathcal{A}$ and $v(B)\in\mathcal{A}$;
	\end{enumerate}
\end{lemma}

Now we turn to the definition of a \cnl-canonical valuation.
\begin{definition}\label{canval1} For each $c$-normal prime theory $\alpha$ and propositional variable $p$ we define a \cnl-canonical valuation $v^c$ as a mapping $PV\mapsto 4\mathcal Q$  satisfying the following expressions:
	\begin{enumerate}[(1)]\itemsep=0pt
		\item $v^c(p)\in\mathcal{D}\Leftrightarrow p \in \alpha$;
		\item $v^c(p)\in\mathcal{A}\Leftrightarrow{\sim^{3}}p \in \alpha$;
	\end{enumerate}
\end{definition}
Define a unique extension of $v^c$ to the set of all formulas in the usual way and denote this extension by $v^c$ as well and prove that extended valuation behaves as expected with respect to the \textit{c}-normal prime theories.

\begin{lemma}[\sc Canonical Valuation Lemma for \cnl]\label{canval2}
	For each $c$-normal prime theory $\alpha$, formula $A$ and extended canonical \cnl-valua\-tion $v^c$ the following statements hold:  
	\begin{enumerate}[(1)]\itemsep=0pt
		\item $v^c(A)\in\mathcal{D}\Leftrightarrow A \in \alpha$,
		\item $v^c(A)\in\mathcal{A}\Leftrightarrow{\sim^{3}}A \in \alpha$.
	\end{enumerate}
\end{lemma}
\begin{proof}
	By induction on the structure of a formula $A$. The base case when $A$ is a propositional variable follows from the definition~\ref{canval1}. Consider the cases for the complex formulas. The induction hypothesis (`IH' in the sequel) claims that lemma is true for their proper subformulas. We will tacitly use the two basic properties of theories, namely, their closure under conjunction and the relation $\vdash_{\cnl}$ throughout the proof.  	
	
	\medskip
	Case $A={\sim}B$.
	
	$v^c({\sim}B)\in\mathcal{D}\Leftrightarrow \text{(by Lemma \ref{eq}) } v^c(B)\notin\mathcal{A}\Leftrightarrow\text{(by IH) }{\sim^{\emph{\tiny3}}}B\notin\alpha  
	\Leftrightarrow$ (by \textit{c}-normality) ${\sim}B\in\alpha.$
	
	$v^c({\sim}B)\in\mathcal{A} \Leftrightarrow v^c(B)\in\mathcal{D}\text{ (by Lemma \ref{eq})}\Leftrightarrow B\in\alpha$ (by IH)
	$\Leftrightarrow {\sim^{\emph{\tiny4}}}B\in\alpha$   (by (a5), (a6)).
	
	\medskip
	Case $A=B\wedge C$. 
	
	$v^c(B\wedge C)\in\mathcal{D}\Leftrightarrow\text{(by Lemma \ref{eq}) }v^c(B)\in\mathcal{D}\text{ and }v^c(C)\in\mathcal{D} 
	\Leftrightarrow $ (by IH) $ B\in\alpha\text{ and }C\in\alpha 
	\Leftrightarrow$ \par $\text{(by (a1), (a2)) }  B\wedge C\in\alpha.$

	$v^c(B\wedge C)\in\mathcal{A}\Leftrightarrow\text{(by Lemma \ref{eq}) } v^c(B)\in\mathcal{A}\text{ or }v^c(C)\in\mathcal{A}  
	\Leftrightarrow$ (by IH) ${\sim^{\emph{\tiny 3}}}B\in\alpha\text{ or }{\sim^{\emph{\tiny 3}}}C\in\alpha  
	\Leftrightarrow$ \par(by \textit{c}-normality) ${\sim}B\notin\alpha\text{ or }{\sim}C\notin\alpha$  $\Leftrightarrow$ (by (b1), (a7)) ${\sim}(B\wedge C)\notin\alpha  
	\Leftrightarrow$ \par $\text{(by \textit{c}-normality) }{\sim^{\emph{\tiny3}}}(B\wedge C)\in\alpha$.
	
	\medskip
	Case $A=B\vee C$. 
	
	$
	v^c(B\vee C)\in\mathcal{D} \Leftrightarrow\text{(by Lemma \ref{eq}) } v^c(B)\in\mathcal{D}\text{ or }v^c(C)\in\mathcal{D}
	\Leftrightarrow$ (by IH) $ B\in\alpha\text{ or }C\in\alpha
	\Leftrightarrow$ \par $\text{(by (a3), (a4), primeness) } B\vee C\in\alpha. 
	$
	
	$
	v^c(B\vee C)\in\mathcal{A}\Leftrightarrow\text{(by Lemma \ref{eq}) } v^c(B)\in\mathcal{A}\text{ and }v^c(C)\in\mathcal{A}
	\Leftrightarrow$ (by IH) ${\sim^{\emph{\tiny 3}}}B\in\alpha\text{ and }{\sim^{\emph{\tiny3}}}C\in\alpha 
	\Leftrightarrow$ 
	\par(by \textit{c}-normality) ${\sim}B\notin\alpha\text{ and }{\sim}C\notin\alpha 
	\Leftrightarrow$ (by (a3), (a4), (a8), (b2), primeness) ${\sim}(B\vee C)\notin\alpha
	\Leftrightarrow$ \par $\text{(by \textit{c}-norm.) }{\sim^{\emph{\tiny 3}}}(B\vee C)\in\alpha.
	$
\end{proof} 
\begin{theorem}[\sc Completeness for \cnl]\label{compl1}
	For any formulas $A$ and $B$, the following holds:
	\[A\vDash_{\cnl} B\Rightarrow A\vdash B\text{ is \cnl-provable}.\]
\end{theorem}	
\begin{proof}
	Suppose $A\vdash B$ is not \cnl-provable. Then, by Lemma~\ref{extension}, there is prime theory $\alpha$ such that $A\in\alpha$ and $B\not\in\alpha$. Then, by Lemma \ref{canval2}, we have that $v^c(A)\in\mathcal{D}$ but $v^c(B)\notin\mathcal{D}$, so $A\not\vDash_{\cnll} B$.
\end{proof} 

\section{Conclusion}

Although we have studied probably the most natural logics of paired cyclic negations, the whole picture is still waiting to be explored.
Even the framework of the four-valued semantics gives some possible directions for the further investigations. Specifically, one can choose other sets of the designated truth values or combine the different collections of designated and anti-designated truth values. On the other hand, alternative definitions of the consequence relation are also possible. To obtain the more abstract results, paired cyclic negations could be put into more general lattice structures, even not necessary finitely based. Having in mind ability to simulate the other negation-like operations, the potential relationships between logical systems appear to be of the main  interest.
\bibliographystyle{eptcs}
\bibliography{biblio}
\end{document}